%% file: arxiv-main.tex
\documentclass[11pt]{article}

\usepackage{array}
\usepackage{amsfonts,amsmath,amssymb,amsthm,mathrsfs}
\usepackage{adjustbox}
\usepackage{tikz}
\usepackage{enumitem}
\usepackage{ifthen}
\usepackage{fullpage}

\usepackage[ruled]{algorithm2e} % For algorithms

\SetAlFnt{\small}
\SetAlCapFnt{\small}
\SetAlCapNameFnt{\small}
\SetAlCapHSkip{0pt}
\IncMargin{-\parindent}

\usepackage{booktabs} % For formal tables
\usepackage{xspace}
\usepackage[pagebackref=true,hidelinks]{hyperref}
\usepackage{bbm}
\usepackage{enumitem}
\usepackage[numbers]{natbib}
\usepackage{multirow}
\usepackage{graphicx}
\usepackage{subfigure}
\usepackage{wrapfig}
\usepackage{cleveref}

\newtheorem{theorem}{Theorem}[section]

\newtheorem{lemma}[theorem]{Lemma}
\newtheorem{corollary}[theorem]{Corollary}

\newtheorem{definition}[theorem]{Definition}

\newenvironment{numberedtheorem}[1]{%
\renewcommand{\thetheorem}{#1}%
\begin{theorem}}{\end{theorem}\addtocounter{theorem}{-1}}

\newenvironment{numberedlemma}[1]{%
\renewcommand{\thetheorem}{#1}%
\begin{lemma}}{\end{lemma}\addtocounter{theorem}{-1}}

\input{macros}

\begin{document}
\title{Menu-size Complexity and Revenue Continuity of Buy-many Mechanisms}
\author{
Shuchi Chawla \\ UW-Madison \\ {\tt shuchi@cs.wisc.edu} \and 
Yifeng Teng \\ UW-Madison \\ {\tt yifengt@cs.wisc.edu} \and 
Christos Tzamos \\ UW-Madison \\ {\tt tzamos@wisc.edu}
}
\date{}

\maketitle
\thispagestyle{empty}

\begin{abstract}
\input{abstract}
\end{abstract}
%\newpage

\setcounter{page}{1}

\input{introduction}

\input{preliminary}

\input{continuity}
\input{menusize}

\section*{Acknowledgement}
The authors would like to thank S. Matthew Weinberg for discussing the relationship between Theorem~\ref{thm:continuity-general} and \cite{psomas2019smoothed}.

% Bibliography
\bibliographystyle{plainnat}
\bibliography{reference}

%\newpage
%\appendix
\appendix
\input{appendix}

\end{document}

%% file: macros.tex
\newcommand{\rev}{\textsc{Rev}}

\newcommand{\brev}{\textsc{Brev}}

\newcommand{\opt}{\textsc{Opt}}
\newcommand{\E}{\mathbb{E}}
\newcommand{\e}{\epsilon}
\newcommand{\mec}{\mathcal{M}}
\newcommand{\optmec}{\mathcal{M}^*}

\newcommand{\adaptrev}{\textsc{BuyManyRev}}
\newcommand{\ts}{\mathbf{t}}

\newcommand{\basicsets}{\mathcal{B}}
\newcommand{\basiccol}{\mathcal{C}}
\newcommand{\col}{C}
\newcommand{\cols}{\mathbf{C}}

\newcommand{\supsize}{\ell}
\newcommand{\mainmec}{\mec_{\dist_{\cols,\ts}}}

\newcommand{\R}{\mathbb{R}}

\newcommand{\dists}{\mathbf{D}}
\newcommand{\dist}{\mathcal{D}}
\newcommand{\simplemecs}{\mathcal{Q}}
\newcommand{\lotto}{\lambda}
\newcommand{\lottos}{\Lambda}
\newcommand{\gap}{\textrm{gap}}

\newcommand{\argmax}{\operatorname{arg\,max}}

%% file: abstract.tex
%!TEX root = ec-main.tex

We study the multi-item mechanism design problem where a monopolist sells $n$ heterogeneous items to a single buyer. We focus on buy-many mechanisms, a natural class of mechanisms frequently used in practice. The buy-many property allows the buyer to interact with the mechanism multiple times instead of once as in the more commonly studied buy-one mechanisms. This imposes additional incentive constraints and thus restricts the space of mechanisms that the seller can use.

In this paper, we explore the qualitative differences between buy-one and buy-many mechanisms focusing on two important properties:
revenue continuity and menu-size complexity. 

Our first main result is that when the value function of the buyer is perturbed multiplicatively by a factor of $1\pm\e$, the optimal revenue obtained by buy-many mechanisms only changes by a factor of $1 \pm \textrm{poly}(n,\e)$. In contrast, for buy-one mechanisms, the revenue of the resulting optimal mechanism after such a perturbation can change arbitrarily.

Our second main result shows that under any distribution of arbitrary valuations, finite menu size suffices to achieve a $(1-\e)$-approximation to the optimal buy-many mechanism. We give tight upper and lower bounds on the number of menu entries
as a function of $n$ and $\e$. On the other hand, such a result fails to hold for buy-one mechanisms as even for two items and a buyer with either unit-demand or additive valuations, the menu-size complexity of approximately optimal mechanisms is unbounded.

%% file: introduction.tex
%!TEX root = ec-main.tex
\section{Introduction}

% revenue optimal mechanisms exhibit many strange properties.
% relatedly, no simple "nice" mechanisms can achieve approximately  
% optimal revenue.  
Revenue optimal mechanisms for multi-item settings are notorious
within mechanism design for exhibiting strange properties. Even in the
simplest possible setting where a seller offers two different items to
a single buyer with additive values, the revenue optimal mechanism may
offer an infinitely large menu of options, each of which is a lottery
or randomized allocation over the items. Furthermore, no ``simple'' or
``nice'' families of mechanisms can achieve any finite approximation
in revenue to the optimal mechanism, even in this two-item additive
setting~\cite{hart2013menu}. In effect, it is impossible to simultaneously achieve
simplicity and near optimality for revenue in multi-item settings. 

% In recent work CTT advocated for studying revenue maximization under
% an additional constraint. they showed ...
Given this impossibility, in recent work, \citet{chawla2019buy} advocated
studying revenue maximization under the so-called ``buy-many''
constraint. Informally, a mechanism is buy-many if the buyer is
allowed to participate in the mechanism any number of times. For
example, a buyer interested in purchasing a subset of items may
purchase the components of this subset individually. Viewing the
mechanism as a function that assigns prices to allocations, the
buy-many constraint is essentially equivalent to a subadditivity
constraint over the prices. 

\citeauthor{chawla2019buy} showed that revenue maximization under the
buy-many constraint is vastly different from unconstrained revenue
maximization. In contrast to the aforementioned impossibility of
simple and near optimal mechanisms in the unconstrained setting, the
optimal buy-many revenue is {\em always} approximated within a factor
of $O(\log n)$ by item pricing.\footnote{\citet{BCKW-JET15} previously showed a
  similar result for unit-demand buyers. \citet{babaioff2018optimal} showed that there is
a revenue gap between unconstrained and buy-many mechanisms even when
buyers' values for different items are independently distributed.} Here $n$ is the number of
items being sold, and item pricing is the mechanism that assigns a
price to every item and lets the buyer purchase any subset at the sum
of the constituent prices.

% buy many is a natural property. 
% A natural question is what
% properties do optimal buy many mechanisms exhibit. We examine two
% such properties in this paper. 

The buy-many constraint is a natural property that most real-world
mechanisms satisfy. All of the simple classes of mechanisms studied in
literature such as item pricings, grand bundle pricing, two part
tariffs, etc. also satisfy this property. As such, buy-many mechanisms
are a worthy object of study. \citeauthor{chawla2019buy} asked whether buy-many
mechanisms exhibit other structural properties that arbitrary
mechanisms do not. In this paper we study two such properties: revenue
continuity and menu size complexity. We now discuss these two
properties, their significance, and our results in detail.

\paragraph{Revenue Continuity.}% Revenue continuity. Definition. Why is it important. 
When sellers invest in improving their products and offerings, they
generally expect that consumers will be willing to spend more on the
products and revenue will follow suit. Surprisingly,
\citet{hart2015maximal} showed that this is not necessarily the case
when the seller has multiple products to sell: it is possible to
construct a distribution over valuations, such that when every
valuation in the support of the distribution weakly increases, the
optimal revenue of the seller strictly decreases. This phenomenon has
come to be known as {\em revenue non-monotonicity}. The same
phenomenon can be exhibited also for buy-many
mechanisms.\footnote{Indeed, consider any example with bounded
  additive valuations for which the unconstrained optimal revenue is
  non-monotone, such as the example in \cite{hart2015maximal}. Add a
  new item to this example, assigning a high value of $H$ to every
  buyer type for this additional item, where $H$ exceeds the maximum
  value that any buyer obtains from any subset of the original
  items. Then the optimal buy-many revenue for the new setting is
  exactly $H$ more than the optimal unconstrained revenue of the
  original setting. In particular, for any mechanism for the original
  setting, by offering the new item within every possible outcome at a
  price of $H$ more, the revenue of the mechanism increases by $H$;
  This new mechanism is buy-many because any two options cost at least
  $2H$, which is more than the value of the buyer from any
  allocation. On the other hand, the optimal revenue cannot increase
  by more than an additive $H$ amount, because given any mechanism for
  the new setting, we can obtain one for the old setting by replacing
  any allocation of the new item by its equivalent monetary value.}

% examples in \cite{hart2015maximal}
% can be converted to apply also to buy-many mechanisms by introducing an additional item 
% for which every consumer type has a high value $H = 1000$. Indeed, the revenue of the optimal mechanism
% increases by exactly $H$ by always selling the extra item at price $H$
% and the mechanism is optimal even in the buy-many setting as 
% it is costly to buy more than a single option (cost at least $2H$, which is higher than value).}

Given this observation, one may ask:
by how much can the revenue decrease if valuations change by only a
tiny bit. Formally, let $\dist$ be a distribution over valuations, and
let $\dist'$ be another distribution obtained by taking each valuation
function in the support of $\dist'$ and changing each component of
this function multiplicatively by some factor in $[1-\e,1+\epsilon]$ for
some small $\epsilon>0$. Can we then show that $\rev_{\dist'}\ge
(1-\epsilon') \rev_{\dist}$ where $\epsilon'$ goes to $0$ as
$\epsilon$ goes to $0$? We call such a property {\em revenue
  continuity}.

% why is continuity interesting/important?
%(Explain algorithmic viewpoint?)

% optimal mechanisms do not exhibit revenue continuity. We show based  
% on psomas et al.
It turns out that the unconstrained optimal revenue does not exhibit
revenue continuity. We use an example from \citet{psomas2019smoothed} to show that for
every $\epsilon>0$, there exist distributions $\dist$ and $\dist'$ of unit-demand valuations
with $\dist'$ being an ``$\epsilon$-perturbation'' of $\dist$ as
described above, such that $\rev_{\dist}=\infty$ whereas
$\rev_{\dist'}<\infty$.

% Revenue continuity of buy many mechanisms. upper and lower bound.
In sharp contrast to the unconstrained setting, we show that buy-many
mechanisms always satisfy revenue continuity, regardless of the
distribution $\dist$. When the seller has $n$ items to sell, the
relative loss in revenue, $\epsilon'$, depends polynomially on $n$ and
$\epsilon$. On the other hand, we show that a polynomial dependence on
$n$ is necessary even when the buyer is unit-demand.

While revenue continuity is inherently interesting, it also has
practical implications. Continuity implies that revenue estimates
established on the basis of market analysis will be robust to errors
in estimating demand. Furthermore, the accuracy of these estimates
will improve directly with reduction in measurement
error. From an algorithmic standpoint, revenue continuity allows discretizing 
the values to their most significant digits through a sufficiently fine multiplicative grid.
This is possible to do without a significant drop in revenue.

We note that our proposed notion of continuity through multiplicative perturbations
differs from other notions proposed in the literature. It is common to consider
additive perturbations as done in ~\cite{rubinstein2018simple,
  babaioff2017menu, daskalakis2012symmetries}; See also
\citet{kothari2019approximation} for a more principled exposition of this approach to bound the change in revenue in mechanisms.
In addition, other notions of distance like total variation and Prokhorov distance have been considered that are useful when learning 
distributions through samples. Continuity results based on these distances typically depend on the range of values achieved by 
consumer valuation functions. In contrast, our notion of multiplicative perturbation is scale invariant and applies
to value distributions of arbitrary support. While it is not possible to guarantee such strong continuity properties in the unconstrained setting even for unit-demand valuations,
in the buy-many setting we show that it holds for arbitrary distributions over general valuation functions.

% comparison against other notions of continuity. e.g in daskalakis et
% al.

\paragraph{Menu Size Complexity.} 
An additional property we consider is the menu-size complexity of buy-many mechanisms.
The menu size of a mechanism, defined as the number of different outcomes the seller offers
to the buyer, has been studied extensively in literature as a measure
of complexity for single-buyer mechanisms (see, e.g., \cite{hart2013menu,
  dughmi2014sampling, babaioff2017menu, gonczarowski2018sample}). It has a direct correspondence
  with the communication complexity of the interaction protocol between the buyer and the seller.

For unconstrained settings, it is known that even for two items and a buyer who is additive or unit-demand,
the menu-size complexity can be infinite~\cite{hart2013menu}. In fact, the same is true for any mechanism
that achieves a bounded approximation to the optimal revenue. 
All known positive results for menu size complexity hold when the valuations are drawn from a product distribution
and approximate mechanisms are considered.
For the optimal mechanism in the case of an additive buyer and two i.i.d. items, the menu-size complexity
of the optimal mechanism is still unbounded~\cite{daskalakis2017strong,gonczarowski2018sample}.

For buy-many mechanisms, the menu size complexity corresponds to
the concept of ``additive menu size'' introduced by 
Hart and Nisan~\cite{hart2013menu}, which corresponds to the number of ``basic''
options a buy-many mechanism offers. 
We show that while the optimal buy-many mechanism might still have unbounded menu-size complexity,
the menu-size required for achieving $1-\epsilon$ approximation to the optimal revenue is always bounded,
even for arbitrary distributions over arbitrary valuations. Indeed, it always suffices to use 
a menu of size $(n/\epsilon)^{2^{O(n)}}$, i.e. doubly exponential in $n$. When one considers only unit-demand
valuations, a menu of size $(n/\epsilon)^{O(n)}$, i.e. singly exponential in $n$, is sufficient.
In both cases, we provide matching lower-bounds showing that any mechanism that achieves $o(\log n)$ approximation
must have doubly exponential (resp. singly exponential) menu size. In fact we show something even stronger---describing a 
mechanism that achieves an $o(\log n)$ approximation in {\em any possible encoding}, requires doubly exponential (resp. singly exponential) number of bits.

% Menu size complexity or description complexity. Definition. Why is
% it important. Additive menu complexity. 

% optimal mechanisms have unbounded menu size complexity. 

% our results for buy many mechanisms. upper and lower bounds. 

%% file: preliminary.tex
\section{Preliminaries}

We study the single-buyer optimal mechanism design problem where the seller has $n$ heterogeneous items to sell to a single buyer. The buyer's type is a valuation function $v:2^{[n]}\to\mathcal{R}_0^+$ which specifies a non-negative value for every possible set of items. We assume the buyer's valuation function is always monotone: for any sets of items $S\subseteq T\subseteq[n]$, $v(S)\leq v(T)$. We study the Bayesian setting where the buyer's valuation function $v$ is drawn from a known distribution $\dist$ over valuation functions. The objective of the seller is to maximize the expected revenue.

By the Taxation Principle, any single-buyer mechanism can be described by a menu of possible outcomes, each of which is a lottery that assigns a price to a randomized allocation. Let $\lottos$ denote the set of lotteries in the menu. Any lottery $\lotto=(x,p)$ is specified by a randomized allocation $x$ over sets of items, and a price $p$ charged for the lottery. Denote by $x_S$ the probability that set $S\in[n]$ is allocated in allocation $x$, and we have $\|x\|_1=1$. We will use $x(\lotto)$ and $p(\lotto)$ to denote the allocation and price of lottery $\lotto$. The buyer's \textit{value} from purchasing menu entry $\lotto$ is defined by $v_\lotto\equiv v(\lotto)\equiv\E_{S\sim x}v(S)$. Since the valuation function $v$ can also be expressed as a vector of length $2^n$, we also use $x\cdot v$ to denote the buyer's value when getting allocation $x$. We assume that the buyer has quasi-linear utility, which means the buyer's \textit{utility} for purchasing lottery $\lotto=(x,p)$ is 
\[u_v(\lotto)=v(\lotto)-p=\E_{S\sim\lotto}v(S)-p.\]

\paragraph{The Buy-one Model.} 
We say that a mechanism $\mec$ defined by a menu of lotteries $\lottos$ is ``\textit{buy-one}'', if the buyer is allowed to purchase only one menu option from $\lottos$. We assume that the buyer is a utility maximizer that always selects the lottery with the highest utility for his realized valuation function $v$, $\argmax_{\lotto\in\lottos} u_v(\lotto)$. 
% interacts with the mechanism once and purchases his favorite menu entry $\lotto\in\lottos$. 

\paragraph{The Buy-many Model.} A buy-many mechanism $\mec$ generated by a set of lotteries $\lottos$ is defined as follows. The buyer can adaptively purchase a (random) sequence of lotteries in $\lottos$, where the buyer can decide which lottery to purchase next after observing the instantiation of the lotteries purchased previously. The buyer gets the union of all items allocated in each step in the sequence of lotteries, and pays the sum of the prices of the lotteries purchased. For any adaptive purchasing algorithm $\mathcal{A}$, let $\lottos_{\mathcal{A}}=(\lotto_{\mathcal{A},1},\lotto_{\mathcal{A},2},\cdots)$ be the random sequence of lotteries purchased by the buyer. The utility of the buyer under this adaptive strategy is specified by 
\[u_v(\lottos_{\mathcal{A}})=\E_{S_i\sim \lotto_{\mathcal{A},i}}\left[v\left(\cup_i S_i\right)\right]-\E_{\mathcal{A}}\sum_{i}p({\lotto_{\mathcal{A},i}}).\]

Any buy-many mechanism can be equivalently described in the form of a buy-one menu. The expected outcome of any adaptive purchasing strategy $\mathcal{A}$ can be described as a single lottery consisting of the allocation $\cup_i (S_i\sim \lotto_{\mathcal{A},i})$ and the price $\E_{\mathcal{A}}\sum_{i}p({\lotto_{\mathcal{A},i}})$. The buy-one menu representing the mechanism is simply the collection of all possible such lotteries. Observe that a buyer offered such a menu cannot improve his utility by purchasing more than one menu options. We accordingly say that the buy-one menu satisfies the buy-many constraint.\footnote{\cite{chawla2019buy} use the term ``Sybil-proof'' for such buy-one mechanisms.}
%Thus, any buy-many mechanism is equivalent to a buy-one mechanism satisfying the buy-many constraint\footnote{This is the adaptively Sybil-proofness defined in \cite{chawla2019buy}.}, which means a buyer cannot improve his utility by interacting with the buy-one mechanism adaptively for multiple times. 
\begin{definition}
A buy-one mechanism $\mec$ defined by lotteries $\lottos$ {\bf satisfies the buy-many constraint} if for {\em every} adaptive buying strategy $\mathcal A$ there exists a cheaper single lottery $\lotto\in\lottos$ dominating it.\footnote{Say an adaptive strategy $\mathcal{A}$ is dominated by a lottery $\lotto$, if there exists a coupling between a random draw $S$ from $\lotto$ and a random union of draws $S'$ from $\lottos_{\mathcal{A}}$, such that $S\supseteq S'$.}%: $p(\lotto)\le \E_{\lottos_{\mathcal A}}\left[\sum_{\lotto'\in\lottos_{\mathcal A}} p(\lotto')\right]$. 
\end{definition}

\paragraph{Revenue.} Define $\rev_\dist(\mec)$ to be the revenue achieved by mechanism $\mec$ when buyer's type is drawn from $\dist$. We also use $\rev_v(\mec)$ to denote the payment of the buyer of type $v$ in mechanism $\mec$. Let $\opt(\dist)$ be the maximum revenue obtained by any truthful mechanism when the buyer's type is drawn from $\dist$. Let $\adaptrev(\dist)$ be the optimal revenue obtained by any buy-many mechanism for a buyer with type drawn from $\dist$.

\paragraph{Menu-size Complexity.} The \textit{menu-size complexity} of a buy-one mechanism $\mec$ is simply defined to be the number of options on its menu. Likewise, the menu-size complexity of a buy-many mechanism is defined to be the number of options on its menu, although the number of different possible allocations made by the mechanism to buyers with different types can be much more numerous. For example, by this definition the menu-size complexity of item pricings is $n$. This definition of menu-size complexity for buy-many mechanisms is similar to Hart and Nisan's ``additive menu-size complexity'', where there are multiple ``basic menu entries'' and the buyer can choose to purchase any subset of them. 

% For buy-one mechanism $\mec$ with corresponding lottery set $\lottos$, the \textit{menu-size complexity} of $\mec$ is defined by $|\lottos|$. Similarly, for buy-many mechanism $\mec$ generated by lottery set $\lottos'$, the menu-size complexity of $\mec$ is defined by $|\lottos'|$. This definition of menu-size complexity for buy-many mechanisms is similar to Hart and Nisan's ``additive menu-size complexity'', where there are multiple ``basic menu entries'' and the buyer can choose to purchase the sum of any subset of them. 

\paragraph{Unit-demand Buyers.} Say a buyer is unit-demand over $n$ items, if for any possible instantiation $v$ of the buyer's type and any set $S$ of items, $v(S)=\max_{i\in S}v(\{i\})$. That is, the buyer is interested in purchasing one item, and his value for any set of items is determined by the best item in the set. For unit-demand buyers, any allocation $x$ can be represented as $x=(x_1,x_2,\cdots,x_n)$ where $x_i$ is the probability that item $i$ is getting allocated. Note that sometimes a buyer with general valuation function can be seen as a buyer that is unit-demand over $2^n$ ``meta items'', where each meta item corresponds to some subset $S\subseteq[n]$ of items.

%% file: continuity.tex
\section{Revenue Continuity}

In this section, we study the ``revenue-continuity'' of buy-many mechanisms. We are interested in understanding the extent to which the optimal revenue changes if the value distribution is perturbed slightly. This extent depends, of course, on the manner in which the value distribution is perturbed. It is known~\cite{rubinstein2018simple,gonczarowski2018sample,kothari2019approximation,brustle2019multi}), for example, that when each valuation function in the support of the distribution gets perturbed additively by some small $\e>0$, the optimal revenue does not change too much multiplicatively. However, when the buyer's values are unbounded, it is natural to ask what happens if the valuation function is perturbed multiplicatively. Formally, we consider the following kind of value perturbation.

\begin{definition}
Call a value distribution $\dist'$ a {\bf$(1\pm\e)$-multiplicative-perturbation} of value distribution $\dist$, if there exists a coupling of $\dist$ and $\dist'$ with coupled draws $v\sim \dist$ and $v'\sim\dist'$, such that $(1-\e)v(S)\leq v'(S)\leq (1+\e)v(S)$ for every item set $S\subseteq [n]$.
\end{definition}

The following surprising result based on an example by \cite{psomas2019smoothed} shows that for buy-one mechanisms, perturbing the value distribution multiplicatively may lead to a significant change in revenue. It is even possible that the optimal revenue is infinite before the perturbation, but finite afterward. Since the theorem was not explicitly stated in \cite{psomas2019smoothed}, we provide a proof in the appendix for completeness.

\begin{theorem}\label{thm:continuity-general}
There exists a distribution $\dist$ over additive value functions over 2 items with $\opt(\dist)=\infty$, such that for any $\e>0$, there exists a distribution $\dist'$ being a $(1\pm\e)$-multiplicative-perturbation of $\dist$ with $\opt(\dist')<\infty$.
\end{theorem}

The main result of this section is that above discontinuity in revenue for buy-one mechanisms does not happen for buy-many mechanisms. In particular, we have the following theorem.

\begin{theorem}\label{thm:continuity-adaptive}
Let $\dist$ be a distribution over arbitrary valuation functions, and $\dist'$ a $(1\pm\e)$-multiplicative-perturbation of $\dist$. Then
\[\adaptrev(\dist')\geq\left(1-O(\e^{1/6}n^{1/2}\log^{1/6}n)\right)\adaptrev(\dist).\]
\end{theorem}

\begin{proof}[Proof of Theorem~\ref{thm:continuity-adaptive}]
The high-level proof idea of the theorem is as follows. We find the optimal buy-many mechanism for distribution $\dist$ and give a discount for every lottery for distribution $\dist'$. For any $v\sim \dist$, it is possible that when perturbed to $v'\sim\dist'$, the buyer switches to purchasing a much cheaper lottery. However, we can prove that such value types do not contribute too much to the optimal revenue achieved by buy-many mechanisms.

We first introduce some notation. For any distribution $\dist$ and $A$ being a set of valuation functions, define $\dist_{|A}$ to be the following value distribution: a draw from $\dist_{|A}$ can be simulated by first sampling $v\sim \dist$, then setting $v$ to be zero function if $v\not\in A$.  Let $\optmec$ be the optimal buy-many mechanism for $\dist$ represented as a buy-one menu. For value distribution $\dist'$, construct a mechanism $\mec'$ such that for each option $(x,p)\in \optmec$, there exists a corresponding lottery $(x,(1-\e')p)\in\mec'$. Here $\e'=\e^{1/6}n^{1/2}\log^{1/6}n$ is parameterized by $\e$ and $n$. The new mechanism $\mec'$ continues to satisfy the buy-many constraint since we reduce the prices of all menu options uniformly by a factor of $1-\e'$, and the price of any adaptive strategy also decreases by a factor of $1-\e'$.

Now we argue that $\rev_{\dist'}(\mec')\geq \left(1-O(\e^{1/3}n^{1/2}\log^{1/3}n)\right)\rev_{\dist}(\optmec)$.
Let set $A$ denote the set of support value function $v$ of $\dist$ such that after being perturbed to $v'\in v(1\pm\e)$,
the buyer switches from purchasing $(x,p)$ in $\optmec$ to $(x',(1-\e')p')$ in $\mec'$, with $p'<(1-\e')p$. For each $v\not\in A$, the revenue contributed from buyer of type $v$ in $\mec'$ decreases by a factor at most $(1-\e')^2$, which is close to 1. Thus to show that the optimal revenue of the new distribution does not drop a lot, it suffices to show that the revenue contributed from buyers with type in $A$ is tiny compared to $\rev_{\dist}(\optmec)$. 

Before value perturbation, suppose that a buyer $v\in A$ chooses to purchase $(x,p)\in\optmec$. Since the buyer is a utility maximizer and prefers not to purchase $(x',p')\in\optmec$, we have
\begin{equation}\label{eqn:beforeperturbsybil2}
v\cdot x-p\geq v\cdot x'-p'.
\end{equation}
After perturbation, the buyer chooses to purchase $(x',p'(1-\e'))$ rather than $(x,p(1-\e'))$. Since $v$ is disturbed to some $v'\in v(1\pm\e)$, the utility of $v'$ purchasing lottery $(x',p'(1-\e'))$ is at most $(1+\e)v\cdot x'-p'(1-\e')$, while the utility of $v'$ purchasing lottery $(x,p(1-\e'))$ is at least $(1-\e)v\cdot x-p(1-\e')$. Thus
\begin{equation}\label{eqn:afterperturbsybil2}
(1+\e)v\cdot x'-p'(1-\e')\geq (1-\e)v\cdot x-p(1-\e').
\end{equation} 
By $(1+\epsilon)*\eqref{eqn:beforeperturbsybil2}+\eqref{eqn:afterperturbsybil2}$, we have $2\e v\cdot x\geq (\e+\e')(p-p')\geq \e'^2p$, thus
\begin{equation}\label{eqn:largevalue}
v([n])\geq v\cdot x\geq\frac{\e'^2}{2\e}p.
\end{equation} 
Since $\frac{\e'^2}{\e}\gg1$, \eqref{eqn:largevalue} implies that in $\optmec$, the buyer purchases a lottery that has price much smaller than his total value. Intuitively, it's possible to raise the prices to extract more revenue from such buyers in $A$. We will exploit this fact to show that these buyer types cannot contribute much revenue in $\optmec$.

Consider the following item pricing $q$, which sets price $q_i=\min_{(x,p)\in\optmec}\frac{p}{\Pr[i\in x]}$ for each item $i\in[n]$. Without loss of generality we can reorder the items in $[n]$ and assume that $q_1\leq q_2\leq \cdots\leq q_n$. The following theorem from \cite{chawla2019buy} shows that there exists a scaled item pricing which obtains a good fraction of the revenue in $\optmec$ from every buyer $v$.
\begin{theorem}\label{thm:buymanyold}
(Restatement of Theorem 1.3 and Theorem 3.2 of \cite{chawla2019buy}) There exists a distribution over scaling factor $\alpha\in[\frac{1}{2n},1]$, such that for any valuation $v$, $\E_{\alpha}[\rev_v(\alpha q)]\geq\frac{1}{2\log 2n}\rev_v(\optmec)$.
\end{theorem}

Henceforth, we will focus on a buyer in $A$ with value $v$, and let $p(v)$ be the price of the menu option bought by this buyer in mechanism $\optmec$. Let $y_i(v)$ denote the probability over the scaling factor $\alpha$ in the random pricing $\alpha q$, that buyer purchases a set of items $S$ with $i\in S$ being the highest-priced item.  Let $A_j\subseteq A$ be the following set of buyer types: $v\in A_j$ if $j\in[n]$ is the largest index such that $y_j(v)>0$, and $q_j>\sqrt{\frac{n\e'^2}{2\e\log 2n}}p(v)$. Let $A'_j\subseteq A$ be the following set of buyer types: $v\in A'_j$ if $j$ is the largest index such that $y_j(v)>0$, while $q_j\leq\sqrt{\frac{n\e'^2}{2\e\log 2n}}p(2)$. Then $A=\bigcup_{j\in[n]}(A_j\cup A'_j)$.

On one hand, for any $v\in A_j$, $y_j(v)>0$ implies that there exists some $\alpha'\in[\frac{1}{2n},1]$ such that under item pricing $\alpha' q$, buyer with type $v$ purchases item $j$. Therefore $v(\{j\})\geq \alpha'q_j\geq\frac{1}{2n}q_j$. Consider the following mechanism $\mec_j$ that only sells item $j$ with deterministic price $\frac{1}{2n}q_j$. For any $v\in A_j$, since the buyer can afford to purchase this item in $\mec_j$, such mechanism can get revenue $\frac{1}{2n}q_j>\sqrt{\frac{\e'^2}{8\e n\log 2n}}p(v)$ from such buyer. Therefore
\[\adaptrev(\dist_{|A_j})\geq\rev_{\dist_{|A_j}}(\mec_j)\geq\sqrt{\frac{\e'^2}{8\e n\log 2n}}\rev_{\dist_{|A_j}}(\optmec),\]
thus 
\begin{equation}\label{eqn:Aj}
\rev_{\dist_{|A_j}}(\optmec)\leq \sqrt{\frac{8\e n\log 2n}{\e'^2}}\adaptrev(\dist_{|A_j})\leq \sqrt{\frac{8\e n\log 2n}{\e'^2}}\adaptrev(\dist).
\end{equation}

On the other hand, for any $v\in A'_j$, notice that $j$ is the most expensive item purchased by $v$ under random item pricing $\alpha q$ for any $\alpha\in [1/2n,1]$. Thus $\E_{\alpha}[\rev_v(\alpha q)]\leq q_j$. By Theorem~\ref{thm:buymanyold},
\begin{equation}\label{eqn:qj}
\frac{1}{2\log 2n}p(v)=\frac{1}{2\log 2n}\rev_v(\optmec)\leq \E_{\alpha}[\rev_v(\alpha q)]\leq q_j,
\end{equation}
Therefore
\begin{equation}\label{eqn:Aj'}
v([n])\geq\frac{\e'^2}{2\e}p(v)\geq \frac{\e'^2}{2\e}\cdot\sqrt{\frac{2\e\log 2n}{n\e'^2}}q_j\geq \frac{\e'^2}{2\e}\cdot\sqrt{\frac{2\e\log 2n}{n\e'^2}}\cdot\frac{1}{2\log 2n}p(v)=\sqrt{\frac{\e'^2}{8\e n\log 2n}}p.
\end{equation}
Here the first inequality is by \eqref{eqn:largevalue}; the second inequality is by the definition of $A'_j$; the third inequality is by \eqref{eqn:qj}. Consider the following mechanism $\mec'_j$ that sells grand bundle $\{1,2,\cdots,n\}$ with deterministic price $\frac{\e'^2}{2\e}\cdot\sqrt{\frac{2\e\log 2n}{n\e'^2}}q_j$. For any $v\in A'_j$, by \eqref{eqn:Aj'} the buyer can afford to purchase the bundle in $\mec'_j$, thus such mechanism can get revenue $\frac{\e'^2}{2\e}\cdot\sqrt{\frac{2\e\log 2n}{n\e'^2}}q_j\geq\sqrt{\frac{\e'^2}{8\e n\log 2n}}p(v)$. Therefore
\[\adaptrev(\dist_{|A'_j})\geq\rev_{\dist_{|A'_j}}(\mec'_j)\geq\sqrt{\frac{\e'^2}{8\e n\log 2n}}\rev_{\dist_{|A'_j}}(\optmec),\]
thus 
\begin{equation}\label{eqn:Aj2}
\rev_{\dist_{|A'_j}}(\optmec)\leq \sqrt{\frac{8\e n\log 2n}{\e'^2}}\adaptrev(\dist_{|A'_j})\leq \sqrt{\frac{8\e n\log 2n}{\e'^2}}\adaptrev(\dist).
\end{equation}
From the two cases of $A_j$ and $A'_j$, by \eqref{eqn:Aj} and \eqref{eqn:Aj2}, we have
\begin{eqnarray*}
\rev_{\dist_{|A}}(\optmec)&=&\sum_{j\in[n]}\rev_{\dist_{|A_j}}(\optmec)+\sum_{j\in[n]}\rev_{\dist_{|A'_j}}(\optmec)\\
&\leq&2n\sqrt{\frac{8\e n\log 2n}{\e'^2}}\adaptrev(\dist)=2n\sqrt{\frac{8\e n\log 2n}{\e'^2}}\rev_\dist(\optmec).
\end{eqnarray*}
Thus the revenue contribution from buyer of type $v\in A$ is small in $\optmec$. Since for $v\not\in A$, the payment from buyer of type $v$ in $\mec'$ decreases by a factor at most $(1-\e')^2$ compared to his payment in $\optmec$, thus
\begin{eqnarray*}
\rev_{\dist'}(\mec')&\geq&\rev_{\dist'_{|\overline{A}}}(\mec')\\
&\geq&(1-\e'^2)\rev_{\dist_{|\overline{A}}}(\optmec)\\
&\geq&(1-\e'^2)\left(1-2n\sqrt{\frac{8\e n\log 2n}{\e'^2}}\right)\rev_{\dist}(\optmec)\\
&=&\left(1-O(\e^{1/6}n^{1/2}\log^{1/6}n)\right)\rev_{\dist}(\optmec)
\end{eqnarray*}
by $\e'=\e^{1/6}n^{1/2}\log^{1/6}n$.

\end{proof}

We notice that the factor above $\left(1-O(\e^{1/6}n^{1/2}\log^{1/6}n)\right)$ not only depends on $\e$, but also depends on the number of items $n$. Such polynomial in $n$ dependency is necessary, due to the following theorem. 

\begin{theorem}
For any $\e>0$ such that $\e n>1$ and any $\delta\in(0,1)$, there exists unit-demand value distribution $\dist$ over $n$ items and its $(1\pm\e)$-multiplicative-perturbation $\dist'$, such that $\adaptrev(\dist)\geq n$, while $\opt(\dist')=\frac{1+\delta}{\e}$.
\end{theorem}

\begin{proof}
Define $c=\frac{1+\delta}{\delta}$. For any unit-demand valuation function $v$, define $v_i=v(\{i\})$ to be the buyer's value for item $i$, $\forall i\in[n]$. For any $k\in[n]$, define valuation function $v^{(k)}$ as follows: $v^{(k)}_k=\frac{1+\e}{\e}c^{k}$, while for any $j\neq k$, $v^{(k)}_j=\frac{1}{\e}c^{k}$. Define the value distribution $\dist$ as follows: with probability $c^{-k}$, the buyer has value function $v=v^{(k)}$, $\forall k\in[n]$; $v\equiv 0$ otherwise. Consider the following item pricing $p$, with price $p_i=c^i$ for each item $i\in[n]$. For any buyer of type $v^{(k)}$, he cannot afford to purchase any item $j>k$, since the item's price $p_j=c^{j}>c^k=v^{(k)}_j$. The utility of purchasing any item $j<k$ is
\[v_j-p_j<v_j=\frac{1}{\e}c^k=v_k-p_k.\]
Thus under such item pricing, the buyer of type $v^{(k)}$ will purchase item $k$. Since item pricing satisfies the buy-many constraint, we have 
\[\adaptrev(\dist)\geq \rev_{\dist}(p)=\sum_{k=1}^{n}c^{-k}p_k=n.\]
On the other hand, define $\dist'$, the $(1\pm\e)$-multiplicative-perturbation of $\dist$ as follows. Each valuation function $v^{(k)}$ is perturbed to value function $\tilde{v}^{(k)}$, where $\tilde{v}^{(k)}_j=\frac{1}{\e}c^k$ for every $j\in[n]$. $\dist'$ is defined as follows: with probability $c^{-k}$, the buyer has value function $v=\tilde{v}^{(k)}$, $\forall k\in[n]$; $v\equiv 0$ otherwise. In $\dist'$, the buyer always has the same value for each item. For such single-parameter buyer, the optimal mechanism is grand bundle pricing with a deterministic price. Thus the optimal mechanism for $\dist'$ is to post a fixed price $\frac{1}{\e}c$ for all of the items, with optimal revenue being 
\[\opt(\dist')=\frac{c}{\e}\sum_{k=1}^{n}c^{-k}<\frac{c}{\e(c-1)}=\frac{1+\delta}{\e}.\]
\end{proof}

%% file: menusize.tex
\section{Menu-size complexity}

In this section, we study the menu-size complexity of buy-many mechanisms for revenue maximization. It is known that for buy-one mechanisms, there exists value distribution $\dist$ even over additive valuation functions over two items such that no mechanism $\mec$ with finite menu-size complexity can get finite approximation in revenue. In particular, $\rev_\dist(\mec)<\infty$, while $\opt(\dist)=\infty$. For buy-many mechanisms, it's still the case that getting optimal revenue may need a mechanism generated by infinitely many options. We show this by observing that for $\dist$ being the value distribution of an additive buyer over two items, where the buyer's value for each item is drawn independently from Beta distribution $B(1,2)$, the optimal buy-one mechanism satisfies the buy-many constraint with infinite menu-size complexity, and no buy-many mechanism with finite menu-size can be optimal. The proof is deferred to the appendix.

\begin{theorem}\label{thm:infinite}
There exists a distribution $\dist$ over additive valuation functions over two items, such that no mechanism $\mec$ generated by finitely many menu options can obtain the optimal revenue achieved by buy-many mechanisms.
\end{theorem}

Suppose that we want a $(1-\e)$-approximation instead of the exactly optimal revenue. Is the infinite menu-size still needed? This is true for general buy-one mechanisms \cite{BCKW-JET15,hart2013menu}. But for buy-many mechanisms we show that finite menu-size is enough, even for arbitrary distributions over arbitrary valuation functions. In Section~\ref{sec:menuub} we study the problem for a unit-demand buyer. We first prove that a menu-size of $(n/\epsilon)^{O(n)}$ is enough for getting a $(1-\e)$-approximation for a unit-demand buyer. Then we observe that such exponential dependency on $n$ is necessary---there exists a unit-demand value distribution such that no mechanism with a description complexity sub-exponential in $n$ can get a $o(\log n)$-approximation in revenue. In Section~\ref{sec:menulb} we show that the above results for a unit-demand buyer can be extended to a general-valued buyer. In particular, we show that for $(1-\e)$-approximation in revenue, a buy-many mechanism with $(n/\epsilon)^{2^{O(n)}}$ options suffices; while for $o(\log n)$-approximation in revenue, a mechanism with description complexity doubly-exponential in $n$ is required.

\input{menusizeub}

\input{menusizelb}

%% file: menusizeub.tex
\subsection{Menu-size Complexity for Unit-demand Buyers}\label{sec:menuub}

In this section, we study the menu-size complexity for unit-demand buyers. We start with a positive result that finite menu-size is enough for $(1-\e)$-approximation in revenue for buy-many mechanisms.
\begin{theorem}\label{thm:unitdemandub}
For any distribution $\dist$ over unit-demand valuation functions, there exists a buy-many mechanism $\mec$ generated by $O((16n)^{6n}\e^{-12n})$ menu entries, such that $\rev_{\dist}(\mec)\geq (1-\e)\adaptrev(\dist)$.
\end{theorem}

\begin{proof}
The proof structure is as follows. First, we show that it is possible to remove all menu entries where the allocation of some items is small. Then we can discretize the allocation probabilities of the remaining menu entries while maintaining the revenue. Such a technique of discretization of allocation probabilities has been used in proving approximation results in menu-size complexity literature such as \cite{kothari2019approximation}.

Let $\optmec$ be the optimal buy-one mechanism satisfying buy-many constraint. The following lemma states that we can ignore all buyers that purchases a menu $(x,p)$ with $0<x_i<\delta$ for some $i\in [n]$ and $\delta=\frac{\e^3}{n^3}$, losing only $O(\sqrt{\e})$ fraction of total revenue.

\begin{lemma}\label{lem:removesmallprob}
There exists mechanism $\mec'$ and a set $A$ of unit-demand buyer types, such that $\rev_{\dist_{|A}}(\mec')\geq (1-3\sqrt{\e})\rev_\dist(\optmec)$, 
and each buyer in $A$ purchases some menu entry $(x',p')\in \mec'$ with $x'_i=0$ or $x'_i\geq\delta=\frac{\e^3}{n^3}$ for every $i\in [n]$. 
\end{lemma} 

\begin{proof}[Proof of Lemma~\ref{lem:removesmallprob}]
For any buyer of value type $v$, denote by $v_i=v(\{i\})$ his value for item $i$, $\forall i\in[n]$.  
For any lottery $\lotto=(x,p)$, let $x_i$ be the allocation of item $i$ in $x$, $\forall i\in[n]$. For each item $i\in[n]$, define $p_i$ to be price of set $\{i\}$. Let $\mec'$ be the following buy-many mechanism: for each menu entry $(x,p)\in\optmec$, there is a menu entry $(x,(1-\sqrt{\e})p)$ in $\mec'$, if $x_i=0$ or $x_i\geq\delta$ for each $i\in[n]$; for every other $x$, the lottery is removed from the menu of $\mec'$. In other words, every lottery with the allocation for each item being not too tiny gets a price discount of factor $1-\sqrt{\e}$ in $\mec'$, while other lotteries are removed. We prove that $\mec'$ satisfies Lemma~\ref{lem:removesmallprob}.

To show this, we identify a set of buyer types where the value is much more than the corresponding payment of the buyer in $\optmec$. Let $A$ be the following set of buyers $v$: For $\lambda=(x,p)$ being the menu entry purchased by buyer of type $v$ in $\optmec$, $\sum_{i:x_i<\delta}x_i(v_i-v_\lambda)>\e p$. Our first observation is that buyer types in this set do not contribute much revenue in the original mechanism. Thus we can ignore the buyers in this set in the later constructions. 
The proof is deferred to the appendix.

\begin{lemma}\label{lem:smallprob}
$\rev_{\dist_{|A}}(\optmec)<\e\rev_\dist(\optmec).$
\end{lemma}

Then we prove that for any $v\not\in A$ that buys $(x,p)$ in $\optmec$, it will purchase a lottery $(x',p')$ that is not much cheaper in the new mechanism $\mec'$.
Let $x^o$ be the allocation that sets the allocation of all items with $x_i<\delta$ to 0 in $x$, and scales every coordinate up uniformly such that $\|x^o\|_1=1$. Let $\lambda^o=(x^o,p^o)$ be the corresponding lottery of allocation $x^o$ in $\optmec$. Define $v^H_\lambda=\sum_{i:x_i\geq\delta}x_iv_i$ to be the buyer's value contributed from items with allocation at least $\delta$, $v^L_\lambda=\sum_{i:x_i<\delta}x_iv_i=v_\lambda-v^H_\lambda$ to be the buyer's value contributed from items with small allocation. We also define $x^H=\sum_{i:x_i\geq\lambda}x_i$ and $x^L=\sum_{i:x_i<\lambda}x_i$ to be the total allocation for items with large and small allocation respectively. Without loss of generality we can assume $x^H+x^L=1$, since otherwise the buyer can purchase the same lottery $(x,p)$ repeatedly until he gets some item to gain more utility. Then $v_{\lambda^o}=\frac{v^H_{\lambda}}{x^H}$. By definition of $A$, for any $v\in A$, $v^L_\lambda\leq x^Lv_\lambda+\e p$, $v^H_\lambda\geq x^Hv_\lambda-\e p$. In $\optmec$, allocation $x^o$ can be simulated by repeatedly buying $\lambda=(x,p)$ until an item with $x_i>\delta$ appears; and in expectation we need to purchase $\frac{1}{x^H}$ copies. Since $\optmec$ satisfies the buy-many constraint, we have $p^o\leq \frac{p}{x^H}$.
%Thus $p^o\geq p+\frac{v^H_{\lambda}}{x^H}-v_{\lambda}\geq (1-\frac{\e}{x^H})p$. 

Also notice that in $\optmec$, the buyer prefers $(x,p)$ to $(x',p')$, thus
\begin{eqnarray}\label{eqn:ineq1}
x'\cdot v-p'\leq x\cdot v-p\leq \frac{v^H_\lambda}{x^H}+\frac{\e}{x^H}p-p
\end{eqnarray}
since $x\cdot v=v_\lambda\leq \frac{1}{x^H}(v^H_\lambda+\e p)$. 
In $\mec'$, the buyer prefers $(x',p'(1-\sqrt{\e}))$ to $(x^o,p^o(1-\sqrt{\e}))$. Thus
\begin{eqnarray}\label{eqn:ineq2}
x'\cdot v-p'(1-\sqrt{\e})\geq x^o\cdot v-p^o(1-\sqrt{\e})=\frac{v^H_{\lambda}}{x^H}-p^o(1-\sqrt{\e})\geq \frac{v^H_{\lambda}}{x^H}-(1-\sqrt{\e})\frac{p}{x^H}.
\end{eqnarray}
Adding the above two inequalities \eqref{eqn:ineq1} and \eqref{eqn:ineq2}, we get
\[\sqrt{\e}p'\geq\left(1+\frac{\sqrt{\e}-1-\e}{x^H}\right)p>\left(1+\frac{\sqrt{\e}-1-\e}{1-\e}\right)p>(\sqrt{\e}-2\e)p\]
by $x^H\geq 1-n\delta=1-\frac{\e^3}{n^2}>1-\e$. Thus $p'\geq (1-2\sqrt{\e})p$ for any $v\not\in A$, combining Lemma~\ref{lem:smallprob} we have
\begin{eqnarray*}
\rev_\dist(\mec')&\geq&\rev_{\dist_{|\overline{A}}}(\mec')\geq(1-\sqrt{\e})(1-2\sqrt{\e})\rev_{\dist_{|\overline{A}}}(\optmec)\\
&\geq&(1-\sqrt{\e})(1-2\sqrt{\e})(1-\e)\rev_\dist(\optmec)>(1-3\sqrt{\e})\rev_\dist(\optmec).
\end{eqnarray*}

\end{proof}

Now we come back to the proof of Theorem~\ref{thm:unitdemandub}. By Lemma \ref{lem:removesmallprob}, we can first find a mechanism $\mec'$ generated by only menus $(x,p)$ with $x_i\geq\delta=\frac{\e^3}{n^3}$, losing only $(1-O(\sqrt{\e}))$ of revenue. Now we construct the following mechanism buy-many $\mec$ generated by finite number of menus as follows. Let $\alpha=\sqrt{\delta}$. For any menu $(x',p')$ in $\mec'$, there exists a menu $(x',(1-\alpha)p')$ in $\mec$, if $x'_i$ is a multiple of $\delta^2$ for every item $i$. In other words, we remove all menu entries in $\mec'$ where the allocation of some item is not a multiple of $\delta^2$.

Suppose that for some buyer $v$, it purchases $(x',p')$ in $\mec'$, and now purchases $(x,p)$ in $\mec$. Since the buyer is utility-maximizer in $\mec'$, we have
\begin{eqnarray}\label{eqn:ineq3}
x'\cdot v-p'\geq x\cdot v-p.
\end{eqnarray}
Suppose that $x''$ is the allocation vector that rounds down every dimension of $x'$ to the closest multiple of $\delta^2$. Let $p''$ be the price of allocation $x''$ in $\mec'$, then $p''\leq p'$. Since $x'_i\geq \delta$ or $x'_i=0$ for each $i\in[n]$, we have $x_i\leq (1+\delta)x''_i$ for the rounded allocation $x''$. Notice that $x'$ can be simulated by purchase a copy of $x'$, then with probability $\delta$ purchase another copy of $x'$, finally remove the redundant items. Since in $\mec$ the buyer prefers to purchase $(x,p)$ rather than purchase $(1+\delta)$ copies of $x''$ and observe that $p''\leq p'$, we have
\begin{eqnarray}\label{eqn:ineq4}
x\cdot v-p(1-\alpha)\geq x'\cdot v-(1+\delta)p''(1-\alpha)\geq x'\cdot v-(1+\delta)(1-\alpha)p'.
\end{eqnarray}
Summing up the above two inequalities \eqref{eqn:ineq3} and \eqref{eqn:ineq4} we get $p\geq (1+\delta-\frac{\delta}{\alpha})p'$. Take $\alpha=\sqrt{\delta}$ we know that the payment of each buyer in $\mec$ decreases by a factor of $(1+\delta-\frac{\delta}{\alpha})(1-\alpha)\geq 1-2\sqrt{\delta}$ compared to his payment in $\mec'$. Since mechanism $\mec$ is generated by menus with allocation for each item being a multiple of $\delta^2$, there are at most $1+\frac{1}{\delta^2}$ possible values for the allocation of each item, thus $(1+\frac{1}{\delta^2})^n=O(n^{6n}\e^{-6n})$ possible menu entries in total. Combined with Lemma~\ref{lem:removesmallprob}, we get a mechanism $\mec$ generated by at most $O(n^{6n}\e^{-6n})$ menu entries with revenue at least $(1-3\sqrt{\e})(1-\sqrt{\delta})\adaptrev(\dist)>(1-4\sqrt{\e})\adaptrev(\dist)$. Theorem~\ref{thm:unitdemandub} follows by replacing $4\sqrt{\e}$ with $\e$.

\end{proof}

The exponential dependency over $n$ in the menu-size is tight, due to the following theorem. The proof of the theorem has the same intrinsic idea of the proof of Theorem 4.4 of \cite{chawla2019buy}, but is simpler and more generalizable. 

\begin{theorem}\label{thm:unitdemandlb}
There exists a distribution $\dist$ over unit-demand valuation functions over $n$ items, such that $\adaptrev(\dist)$ is a factor of $\Omega(\log n)$ larger than the revenue of any mechanism that can be described using $o\left(\frac{1}{\sqrt{n}}2^{n^{1/4}}\right)$ number of bits.
\end{theorem}

\begin{proof}[Proof of Theorem~\ref{thm:unitdemandlb}]
The proof idea is as follows. Any set of mechanisms that can be described by exponential number of bits contains at most doubly-exponential number of mechanisms. We construct a distribution $\dists$ over value distributions, such that for any fixed mechanism, the probability that the mechanism can get $o(\log n)$-approximation in revenue is doubly-exponentially small when the value distribution $\dist$ is drawn from $\dists$. Then the theorem is shown by applying union bound.

The way we construct a distribution $\dists$ over value distributions is as follows. Each distribution $\dist\sim\dists$ is defined by a vector $\ts=(t_1,\cdots,t_N)\in[1,poly(n)]^N$ of real numbers and a set of $N=2^{n^{1/4}}$ item sets $\basicsets=\{S_1,\cdots,S_N\}$: each set has size $\sqrt{n}$, and any pairwise intersection of two sets has size $\leq n^{1/4}$. For each set $S_i$, there is a buyer being unit-demand over $S_i$ and has value $t_i$ for getting an item. 
%We observe that a mechanism that sells a random item in $S_i$ with price $p_i=\frac{1}{2}t_i$ satisfies the buy-many constraint. 
$\dists$ is constructed by always taking the same set of sets $\basicsets$, and draw each $t_i\in\ts$ randomly from a geometric distribution.

Throughout the proof, we will call any item set of size $\sqrt{n}$ a ``\textit{basic set}''. We first construct a set of $N=2^{n^{1/4}}$ basic sets, such that the pairwise intersection is small. The proof is deferred to the appendix.

\begin{lemma}\label{lem:neighbor}
There exists a collection $\basicsets$ of $N=2^{n^{1/4}}$ sets of size $\sqrt{n}$, such that any two different sets $S,S'\in\basicsets$, $|S\cap S'|\leq n^{1/4}$.
\end{lemma}
Let $\basicsets=\{S_1,S_2,\cdots,S_N\}$ be the sets constructed by Lemma~\ref{lem:neighbor}. For any set $S_i\in\basicsets$ and $t\in \R$, let $v_{S_i,t}$ denote the following unit-demand valuation function: for any set of items $S\subseteq [n]$, $v_{S_i,t}(S)=t$ if $S\cap S_i\neq\emptyset$; otherwise $v_{S_i,t}(S)=0$. That is, the buyer has value $t$ for accepting any item in $S_i$. Define $H=\frac{1}{3}n^{1/4}$. For any $\ts=(t_1,t_2,\cdots,t_N)\in[1,H]^N$ being a value vector of length $N$, define value distribution $\dist_{\basicsets,\ts}$ to be a uniform distribution over value functions $(v_1,v_2,\cdots,v_n)=(v_{S_1,t_1},v_{S_2,t_2},\cdots,v_{S_N,t_N})$. 

Consider the following mechanism $\mec_{\dist_{\basicsets,\ts}}$: for each $v_i=v_{S_i,t_i}$, there is a corresponding menu entry $(x_i,p_i)\in\mec_{\dist_{\basicsets,\ts}}$, where $x_i$ allocates each item in $S_i$ with probability $\frac{1}{\sqrt{n}}$, and $p_i=\frac{1}{2}t_i$. On one hand, each buyer of type $v_i$ can afford to purchase lottery $(x_i,p_i)$ with utility $\frac{t_i}{2}$. On the other hand, a buyer of type $v_i$ purchasing another lottery $(x_j,p_j)$ will let him get an item in $S_i$ with probability $\frac{1}{n^{1/4}}$, since $|S_i\cap S_j|\leq n^{1/4}$. Thus his utility of purchasing $(x_j,p_j)$ is at most
\[\frac{1}{n^{1/4}}t_i-p_j=\frac{1}{n^{1/4}}t_i-\frac{1}{2}p_i\leq\frac{1}{n^{1/4}}\cdot\frac{1}{3}n^{1/4}-\frac{1}{2}\cdot 1<0,\]
since $t_i,t_j\in[1,\frac{1}{3}n^{1/4}]$. Notice that for buyer of type $v_i$, the utility of purchasing any collection of menus is subadditive over the lotteries, since the buyer's valuation function is subadditive. Thus for any adaptive strategy that purchases $\alpha_j$ copies of lottery $(x_j,p_j)$ in expectation, $\forall j\in[\ell]$, the utility of buyer $v_i$ is at most the sum of his utility on purchasing $\alpha_j$ copies of lottery $(x_j,p_j)$ for each $j$. Since his utility for purchasing any lottery other than $(x_j,p_j)$ is negative, no adaptive strategy purchasing lottery other than $(x_i,p_i)$ can be optimal. Therefore, the optimal strategy for any buyer of type $v_i$ is to purchase one copy of lottery $(x_i,p_i)$. Then $\adaptrev(\dist_{\basicsets,\ts})=\frac{1}{2N}\sum_{i=1}^{N}t_i$.

Let $\simplemecs$ be arbitrary class of mechanisms with $|\simplemecs|\leq 2^{N/H^2}$, and set $c=\frac{1}{12}\log H$. It suffices to show that there exists a value vector $\ts$, $\rev_{\dist_{\basicsets,\ts}}(\mec)<\frac{1}{c}\adaptrev(\dist_{\basicsets,\ts})$, for every $\mec\in\simplemecs$.

Suppose that we generate a random $\ts=(t_1,t_2,\cdots,t_N)\in[1,H]$, where each $t_i$ is independently drawn from the following truncated geometric distribution: $\Pr[t_i=2^{a}] = \frac{2^{-a}}{1-H^{-1}}$ for $1\leq a\leq \log H$. Fix a mechanism $\mec\in\simplemecs$. We study the probability that $\rev_{\dist_{\basicsets,\ts}}(\mec)\geq\frac{1}{c}\adaptrev(\dist_{\basicsets,\ts})$, over the randomness of $\ts$.

For the fixed mechanism $\mec\in\simplemecs$, define $h_i$ to be the payment of a buyer with type $v_{S_i,t_i}$ in mechanism $\mec$. Thus, to bound the probability that $\rev_{\dist_{\basicsets,\ts}}(\mec)<\frac{1}{c}\adaptrev(\dist_{\basicsets,\ts})$, we only need to bound the probability that $\Pr_{\ts}[\frac{1}{N}\sum_i h_i\geq\frac{1}{2cN}\sum_i t_i]=\Pr_{\ts}[\sum_i h_i\geq\frac{1}{2c}\sum_i t_i]$.

For any set $S_i\in\basicsets$, define value distribution $\dist_i$ as follows. A draw $v\sim \dist_i$ can be simulated by first draw $t_i$ according to the truncated geometric distribution: $\Pr[t_i=2^{a}] = \frac{2^{-a}}{1-H^{-1}}$ for $1\leq a\leq \log H$; then return value function $v=v_{S_i,t_i}$. In other words, for buyer type $v\sim \dist_i$, $v$ is always demanding the same set of items, but the value is drawn from some equal-revenue distribution. Notice that selling any lottery to this buyer is equivalent to sell fractional amounts of items in $S_i$ to a single-parameter buyer. The optimal revenue from this single-parameter buyer from any mechanism, over the randomness of $t_i$, can be achieved by a pricing the grand bundle of all items at a deterministic price. Such optimum is obtained by selling the grand bundle at price $H$, and the optimal revenue is $\frac{1}{1-H^{-1}}<2$. Thus $h_i$ is a random variable with expectation 
\[\E h_i=\rev_{\dist_i}(\mec)<2.\]
On the other hand, $\E t_i=\frac{\log H}{1-H^{-1}}$. Thus
\begin{eqnarray*}
%\Pr_{\ts}\left[\sum_i h_i \right. & \geq& \left. \frac{1}{2c}\sum_i t_i\right]\\
\Pr_{\ts}\left[\sum_i h_i \geq \frac{1}{2c}\sum_i t_i\right]
&\leq&\Pr\left[\sum_i h_i\geq\frac{1}{4c}\E\left[\sum_i t_i\right]\right]+\Pr\left[\sum_i t_i<\frac{1}{2}\E\left[\sum_i t_i\right]\right]\\
&=&\Pr\left[\sum_i h_i\geq\frac{N}{4c}\cdot\frac{\log H}{1-H^{-1}}\right]+\Pr\left[\sum_i t_i<\frac{N}{2}\cdot\frac{\log H}{1-H^{-1}}\right]\\
&\leq&\exp\left(-\frac{2\left(\frac{N}{4c}\cdot\frac{\log H}{1-H^{-1}}-2N\right)^2}{N H^2}\right)+\exp\left(-\frac{2\left(\frac{N}{2}\cdot\frac{\log H}{1-H^{-1}}\right)^2}{N H^2}\right)\\
&<&2\exp\left(-\frac{2N}{H^2}\right)<2^{-N/H^2}.
\end{eqnarray*}
Here the first line is by union bound. The second line is by $\E t_i=\frac{\log H}{1-H^{-1}}$. The third line is by Hoeffding's inequality and observe that $0\leq h_i,t_i\leq H$. The last line is by $c=\frac{1}{12}\log H$. Then by union bound, since there are only $2^{N/H^2}$ mechanisms in $\simplemecs$, there exists $\ts$ such that no mechanism in $\simplemecs$ can get revenue at least $\frac{1}{c}\adaptrev(\dist_{\basicsets,\ts})$. This proves Theorem~\ref{thm:unitdemandlb} since $c=\frac{1}{12}\log H=\Omega(\log n)$, and any mechanism set of size $>2^{N/H^2}$ cannot be described by less than $\frac{N}{H^2}=\frac{9}{\sqrt{n}}2^{n^{1/4}}$ bits.
\end{proof}

%% file: menusizelb.tex
%!TEX root = ec-main.tex
\subsection{Menu-size Complexity for General-valued Buyers}\label{sec:menulb}

We know that a buyer with general value distribution can be seen as unit-demand over $2^n$ possible sets. The proof of Theorem~\ref{thm:unitdemandub} completely goes through when the buyer is unit-demand over $2^n$ sets instead of unit-demand over $n$ items, thus we have the following theorem. While the theorem is not a direct corollary of Theorem~\ref{thm:unitdemandub}, the proof is almost identical to the unit-demand case, thus we skip the proof of the theorem.

\begin{theorem}\label{thm:generalub}
For any distribution $\dist$ over arbitrary valuation functions, there exists a mechanism $\mec$ generated by $O((16\cdot 2^n)^{6\cdot 2^n}\e^{-12\cdot 2^n})$ menu entries, such that $\rev_{\dist}(\mec)\geq (1-\e)\adaptrev(\dist)$.
\end{theorem}

The menu-size we get for $(1-\e)$ approximation in revenue is doubly-exponential in $n$. We can show that such doubly-exponential menu-size dependency on $n$ is the best possible for $(1-\e)$-approximation in revenue. We show that getting $o(\log n)$-approximation in revenue may require a mechanism with description complexity that is doubly-exponential in $n$, even when the buyer has XOS valuation functions. This improves the lower bound in \cite{chawla2019buy} exponentially for general-valued buyer.
\begin{theorem}\label{thm:lb}
There exists a distribution $\dist$ over XOS valuation functions over $n$ items, such that $\adaptrev(\dist)$ is a factor of $\Omega(\log n)$ larger than the revenue of any mechanism that can be described using $o\left(\frac{1}{\sqrt{n}}2^{2^{n^{1/4}/4}}\right)$ number of bits.
\end{theorem}
The proof idea is as follows. Similar to the proof for unit-demand case, we notice that the number of mechanisms that can be described by doubly-exponential number of bits is at most triply-exponential in $n$. We construct a distribution $\dists$ over distributions of XOS value functions, such that for any fixed mechanism, the probability that the mechanism can get $o(\log n)$-approximation in revenue is triply-exponentially small when the value distribution $\dist$ is drawn from $\dists$. 

The XOS value functions we consider for the buyer are of the form $v(S) \propto \max_i |S_i \cap S|$, for a collection of sets $S_i$ that have large size $|S_i| = \sqrt{n}$ but small pairwise intersections $|S_i \cap S_j| \le n^{1/4}$. These valuation functions can be thought of as (approximately) unit-demand over the exponentially-many bundles of items $S_i$. Based on this observation, we adapt the lower-bound proof for unit-demand buyers to general-value functions that are approximately unit-demand over bundles of items. This allows us to strengthen the lower-bound from exponential to doubly-exponential. See the appendix for the complete proof.

%% file: appendix.tex
\section{Omitted proofs}

\begin{numberedtheorem}{\ref{thm:continuity-general}}
\cite{psomas2019smoothed} There exists a distribution $\dist$ over additive value functions over 2 items with $\opt_{\dist}=\infty$, such that for any $\e>0$, there exists a distribution $\dist'$ being a $(1\pm\e)$-multiplicative-perturbation of $\dist$ with $\opt_{\dist}<\infty$.
\end{numberedtheorem}

\begin{proof}[Proof of Theorem~\ref{thm:continuity-general}]
We will use the following three theorems from previous literature. The first two theorems show that there exists a distribution $\dist$ over additive value functions, such that for any $v=(v_1,v_2)\sim \dist$, $\frac{1}{2}v_1\leq v_2\leq 2v_1$, and $\opt(\dist)=\infty$.
\begin{theorem}\label{thm:7-1}
(Restatement of Proposition 7.1 of \cite{hart2013menu}) Let $(g_n)_{n=1}^{\infty}$ be a sequence in $[0,1]^2$ such that
\[\gap_n:=\min_{1\leq j<n}(g_n-g_j)\cdot g_n>0\]
for all $n\geq 1$. Then there exists a sequence $(t_n)_{n=1}^{\infty}$ of positive real number and a buyer's value distribution $\dist$ such that
\begin{enumerate}
    \item For any $v\sim \dist$, there exists integer $n$ such that $v=t_ng_n$;
    \item $\frac{\opt(\dist)}{\brev(\dist)}>\frac{1}{2}\sum_{n=1}^{\infty}\frac{\gap_n}{\|g_n\|_1}$, and $\brev(\dist)<\infty$.
\end{enumerate}
\end{theorem}

\begin{theorem}\label{thm:7-5}
(Slightly stronger version of Proposition 7.5 of \cite{hart2013menu}) There exists an infinite sequence of points $(g_n)_{n=1}^{\infty}$ in $[0,1]^2$ with $\|g_n\|_2\leq 1$ such that $\gap_n=\Omega(n^{-6/7})$. Furthermore, for any $g_n=(g_{n,1},g_{n,2})$ in the sequence, $\frac{1}{2}g_{n,1}\leq g_{n,2}\leq 2g_{n,1}$.
\end{theorem}
The theorem is slightly stronger than Proposition 7.5 of \cite{hart2013menu}, by adding the constraint $\frac{1}{2}g_{n,1}\leq g_{n,2}\leq 2g_{n,1}$. This can be proved by modifying the original proof of Proposition 7.5 of \cite{hart2013menu} as follows. In the original proof, \cite{hart2013menu} shows that the sequence of points $(g_n)$ composed of a sequence of ``shells'', each containing multiple points, satisfy $\gap_n=\Omega(n^{-6/7})$. All points $g_n$ in the $N$-th shell are of length $\|g_n\|_2=\sum_{\ell\leq N}\ell^{-3/2}/\sum_{\ell< \infty}\ell^{-3/2}$, and each shell $N$ contains $N^{3/4}$ different points in it with the angle between any two of them being at least $\Omega(N^{-3/4})$. \cite{hart2013menu} did not specify how to choose these points in each shell, but we notice that it's possible to restrict all points $g_n=(g_{n,1},g_{n,2})$ chosen to satisfy $\frac{1}{2}g_{n,1}\leq g_{n,2}\leq 2g_{n,1}$, and still make the angle between any two points in a shell being $\Omega(N^{-3/4})$. This proves Theorem~\ref{thm:7-5}.

The following theorem from \cite{psomas2019smoothed} shows that for any value distribution, if each possible value shifts uniformly by a small square of length $\delta$, the optimal revenue is always upper bounded by $O(\frac{1}{\delta^2}\brev)$. 

\begin{theorem}\label{thm:4-1}
(Theorem 4.1 of \cite{psomas2019smoothed}). For any distribution $\dist$, let $\dist'$ be the following value distribution. A draw $v'\sim\dist'$ can be simulated by first draw $v=(v_1,v_2)\sim \dist$, then set $v'=v+(\delta_1,\delta_2)$, where $\delta_1,\delta_2$ are drawn independently from uniform distribution $U[0,\delta\max(v_1,v_2)]$. Then $\rev(\dist')\leq \frac{\pi}{\delta^2}\brev(\dist')$.
\end{theorem}

Apply Theorem~\ref{thm:4-1} to a distribution $\dist$ over additive value functions, such that for any $v=(v_1,v_2)\sim \dist$, $\frac{1}{2}v_1\leq v_2\leq 2v_1$, and $\opt(\dist)=\infty$. The square-shift with $\delta=2\e$ under the assumption that for any $v=(v_1,v_2)\sim \dist$, $\frac{1}{2}v_1\leq v_2\leq 2v_1$ can be captured by $(1\pm\e)$-multiplicative-perturbation. Thus for any $\e>0$, there exists a distribution $\dist'$ being a $(1\pm\e)$-multiplicative-perturbation of $\dist$ with $\opt(\dist)<\infty$.

\end{proof}

\begin{numberedtheorem}{\ref{thm:infinite}}
There exists a distribution $\dist$ over additive valuation functions over 2 items, such that no mechanism $\mec$ generated by finitely many menus can get optimal revenue achieved by buy-many mechanisms.
\end{numberedtheorem}

\begin{proof}[Proof of Theorem~\ref{thm:infinite}]
Let $\dist$ be the value distribution of an additive buyer over 2 items, where the buyer's value for each item is drawn independently from Beta distribution $B(1,2)$. That is, the buyer's value $v_i$ for each item $i\in\{1,2\}$ is drawn from $[0,1]$ with density $f(v_i)=2-2v_i$. In previous literature Daskalakis et al. \cite{daskalakis2017strong} characterized the unique optimal mechanism for this setting with infinite menu-size as follows. Define $x_0=y_0\approx 0.0618$ and $p^*\approx0.5535$ be constants calculated in \cite{daskalakis2017strong}. The value space $[0,1]^2$ is divided into 4 regions: let $\mathcal{A}=\{(v_1,v_2):v_1\in[0,x_0),\ v_2\in[\frac{2-3v_1}{4-5v_1},1]\}$, $\mathcal{B}=\{(v_1,v_2):v_2\in[0,y_0),\ v_1\in[\frac{2-3v_2}{4-5v_2},1]\}$, $\mathcal{W}=\{(v_1,v_2)\in[x_0,1]\times[y_0,1]:v_1+v_2\geq p^*\}$; $Z=[0,1]^2\setminus(\mathcal{A}\cup \mathcal{B}\cup\mathcal{W})$. For buyer type in each region, the allocation and payment in the optimal mechanism $\optmec$ are defined as follows.
\begin{itemize}
    \item If $(v_1,v_2)\in Z$, the buyer gets no items and pays 0.
    \item If $(v_1,v_2)\in \mathcal{A}$, the buyer gets item 1 with probability $\frac{2}{(4-5v_1)^2}$, item 2 with probability 1, and pays $\frac{15v_1^2-20v_1+8}{(4-5v_1)^2}$.
    \item If $(v_1,v_2)\in \mathcal{B}$, the buyer gets item 2 with probability $\frac{2}{(4-5v_2)^2}$, item 1 with probability 1, and pays $\frac{15v_2^2-20v_2+8}{(4-5v_2)^2}$.
    \item If $(v_1,v_2)\in \mathcal{W}$, the buyer gets both item 1 and item 2, and pays $p^*\approx0.5535$.
\end{itemize}
The following lemma shows that the optimal buy-one mechanism stated above satisfies the buy-many constraint.
\begin{lemma}\label{lem:betasybil}
The optimal buy-one mechanism $\optmec$ for distribution $\dist$ satisfies the buy-many constraint.
\end{lemma}
\begin{proof}[Proof of Lemma~\ref{lem:betasybil}]
We argue that given the menu in $\optmec$, the optimal adaptive strategy for buyer of any type is to purchase a single lottery: his corresponding lottery in the buy-one mechanism $\optmec$. For any buyer with value $(v_1,v_2)$, if his optimal adaptive strategy only purchases one lottery, then the lottery must be his allocation in $\optmec$. Otherwise, since any lottery in $\optmec$ allocates one of the two items deterministically, without loss of generality in his optimal adaptive strategy we assume the buyer first purchases a lottery $(x,p)$ with $x=(1,a)$ for some $a\in[0,1)$ that allocates item 1 surely. If he did not get both items, then purchase lottery $(x',p')$ with $x'=(0,1)$. Such adaptive strategy is the only possible optimal strategy that does not purchase a single lottery. In such adaptive strategy, the buyer always gets both items, with expected payment $p+(1-a)p'$. Notice that in $\optmec$, $a<\frac{2}{(4-5*0.0618)^2}<0.147$; while $p\geq p'=0.5$. Thus the expected payment of this adaptive strategy is at least $0.5+(1-0.147)*0.5>0.5535=p^*$, which means that compared to this strategy, it's better to purchase both items with price $p^*$ directly. This means that $\optmec$ satisfies the buy-many constraint.
\end{proof}
Back to the proof of Theorem~\ref{thm:infinite}. By above lemma, for distribution $\dist$, we have $\adaptrev(\dist)=\opt(\dist)$. By Theorem 2 of \cite{gonczarowski2018bounding}, any buy-one mechanism with menu-size $o(\e^{-1/4})$ cannot get revenue $\opt(\dist)-\e$, thus no buy-one mechanism with finite menu-size complexity can get optimal revenue for $\dist$. Now we argue that any adaptively buy-many mechanism generated by finite number of lotteries also cannot get optimal revenue. 

Consider any mechanism generated by a finite set of lotteries $\lottos=\{\lotto_1,\lotto_2,\cdots,\lotto_N\}$. The optimal adaptive strategy of any buyer has the following form. The buyer starts with no item at hand, and purchases a lottery $\lotto_{\emptyset}\in\lottos$. Whenever the buyer has a set of items $S\subseteq\{1,2\}$ at hand, he has a lottery $\lotto_{S}\in\lottos$ to purchase, which only depends on the set of items he has and is irrelevant with the purchase history. Since there are only 4 possible states of items (i.e., $\emptyset$, $\{1\}$, $\{2\}$ and $\{1,2\}$), there are only $N^4$ possible optimal adaptive strategies for all of the buyer types. The expected outcome of each adaptive strategy corresponds to a lottery. Therefore the adaptively buy-many mechanism generated by $\lottos$ corresponds to a buy-one mechanism with at most $N^4$ menu entries. The theorem is proved since such a buy-one mechanism with finite menu-size cannot be optimal.

\end{proof}

\begin{numberedlemma}{\ref{lem:smallprob}}
For any distribution $\dist$, let $\optmec$ be the optimal buy-one mechanism that satisfies the buy-many constraint. Let $A$ be the following set of buyer types $v$: for $\lambda=(x,p)$ being the menu entry purchased by buyer of type $v$ in $\optmec$, $\sum_{i:x_i<\delta}x_i(v_i-v_\lambda)>\e p$. Then
\[\rev_{\dist_{|A}}(\optmec)<\e\rev_\dist(\optmec).\]
\end{numberedlemma}

\begin{proof}[Proof of Lemma~\ref{lem:smallprob}]

For any $v\in A$, if $\lambda=(x,p)$ is the lottery purchased by the buyer in $\optmec$, there exists $i\in[n]$ such that $x_i(v_i-v_\lambda)>\frac{\e}{n}p$ and $x_i<\delta$. Define $A_i$ to be the set of buyer type $v$ that purchases $\lambda=(x,p)$ with $x_i(v_i-v_\lambda)>\frac{\e}{n}p$ and $x_i<\delta$, and assume by way of contradiction $\rev_{\dist_{|A_i}}(\optmec)>\frac{1}{n}\e\rev_\dist(\optmec)$. For any $v\in A_i$, since the utility-maximizing buyer purchases $\lambda=(x,p)$ in $\optmec$ rather than $(\{i\},p_i)$, we have
\[v_\lambda-p\geq v_i-p_i.\]
Thus $p_i\geq p+v_i-v_\lambda>\frac{\e}{nx_i}p>\frac{\e}{n\delta}p$. Since purchasing set $\{i\}$ can be simulated by repeatedly purchasing lottery $(x,p)$ until item $i$ appears, which needs to purchase $\frac{1}{x_i}$ copies of the menu, we have $p_i\leq \frac{p}{x_i}$. By definition of set $A_i$, $x_iv_i>\frac{\e}{n}p\geq\frac{\e}{n}x_ip_i$, thus $v_i>\frac{\e}{n}p_i$. Consider the following mechanism $\mec_i$: only sells item $i$ at price $\frac{\e}{n}p_i$. Then every buyer in $A_i$ can afford to purchase it, while paying $\frac{\e}{n}p_i>\frac{\e^2}{n^2\delta}p=\frac{n}{\e}p$. Thus the payment of each buyer in $A_i$ raises by a factor of $\frac{n}{\e}$. By assumption $\rev_{\dist_{|A_i}}(\optmec)>\frac{1}{n}\e\rev_\dist(\optmec)$, we know that such mechanism $\mec_i$ will achieve revenue strictly more than $\frac{1}{n}\e\rev_\dist(\optmec)\cdot\frac{n}{\e}=\rev_\dist(\optmec)$ from buyers in $A_i$. Since $\mec_i$ is an item pricing mechanism, it satisfies the buy-many constraint. This contradicts the optimality of $\optmec$. 

Thus $\rev_{\dist_{|A_i}}(\optmec)\leq\frac{1}{n}\e\rev_\dist(\optmec)$, and 
\[\rev_{\dist_{|A}}(\optmec)=\sum_{i\in[n]}\rev_{\dist_{|A_i}}(\optmec)\leq \e\rev_\dist(\optmec).\]
\end{proof}

\begin{numberedlemma}{\ref{lem:neighbor}}
There exists a collection $\basicsets$ of $N=2^{n^{1/4}}$ sets of size $\sqrt{n}$, such that any two different sets $S,S'\in\basicsets$, $|S\cap S'|\leq n^{1/4}$.
\end{numberedlemma}
\begin{proof}[Proof of Lemma~\ref{lem:neighbor}]
For any $S$ of size $\sqrt{n}$, let $N_y$ be the number of item sets $S'$ of size $\sqrt{n}$ such that $|S\cap S'|=y$, $\forall x\leq n^{1/4}$. Since such a set $S'$ can be determined by first choose $y$ items in $S$, then choose $n^{1/2}-y$ items out of $S$, thus $N_y=\binom{n^{1/2}}{y}\binom{n-n^{1/2}}{n^{1/2}-y}$. Notice that $N_y=N_{y+1}\cdot\frac{(n-2n^{1/2}+y+1)(y+1)}{(n^{1/2}-y)^2}\geq N_{y+1}\cdot(y+1)$ for $y\geq 1$. Thus for a random set $S'$ of size $\sqrt{n}$, the probability that $|S'\cap S|>n^{1/4}$ is 
\[\frac{\sum_{y=n^{1/4}+1}^{n^{1/2}}N_x}{\sum_{x=0}^{n^{1/2}}N_x}<\frac{\sum_{y=n^{1/4}+1}^{n^{1/2}}N_1\frac{1}{y!}}{N_1}<\frac{1}{n^{1/4}!}.\]
Then if we randomly select $N=2^{n^{1/4}}$ sets of size $\sqrt{n}$, by union bound, the probability that there exists two sets with intersection size $>n^{1/4}$ is at most $\frac{N^2}{n^{1/4}!}<1$ for large enough $n$. Thus, there exists a collections of $N$ sets of size $\sqrt{n}$ such that the size of pairwise intersection is at most $n^{1/4}$.
\end{proof}

\begin{numberedtheorem}{\ref{thm:lb}}
There exists a distribution $\dist$ over XOS valuation functions over $n$ items, such that $\adaptrev(\dist)$ is a factor of $\Omega(\log n)$ larger than the revenue of any mechanism that can be described using $o\left(\frac{1}{\sqrt{n}}2^{2^{n^{1/4}/4}}\right)$ number of bits.
\end{numberedtheorem}
\begin{proof}[Proof of Theorem~\ref{thm:lb}]
Let's start with some notations used in this proof. Define $N=2^{n^{1/4}}$, $H=\frac{1}{3}n^{1/4}$, $m=\sqrt{N}$, $\supsize=2^{N^{1/4}}$. Throughout the proof, we will call any item set of size $\sqrt{n}$ a ``\textit{basic set}'', and any collection of $m=\sqrt{N}$ basic sets a ``\textit{basic collection}''. We first construct a set $\basicsets$ of $N=2^{n^{1/4}}$ basic sets through Lemma~\ref{lem:neighbor}, such that the pairwise intersection is small. Apply the above lemma to the collection of basic sets $\basicsets$ constructed in Lemma~\ref{lem:neighbor}, we can also construct a set of $\supsize=2^{N^{1/4}}$ basic collections $\cols$, such that for any two different basic collections $\col_i,\col_j\in\cols$, $|\col_i\cap \col_j|\leq N^{1/4}$. For any basic collection $\col\in\basiccol$ and $t\in \R$, let $v_{\col,t}$ denote the following XOS valuation function: for any set of items $S\subseteq[n]$, $v_{\col,t}(S)=\frac{t}{\sqrt{n}}\max_{S'\in\col}|S\cap S'|$. In other words, such buyer has value $t$ if he is able to get one of the sets in $\col$, and the value decreases linearly with respect to the maximum portion he could get from any set in $\col$.

For any $\ts=(t_1,\cdots,t_{\supsize})\in\R^\supsize$ being a value vector of length $\supsize$, and $\cols=(\col_1,\cdots,\col_{\supsize})$ being the basic collections we constructed above, define value distribution $\dist_{\cols,\ts}$ to be the uniform distribution over valuation functions $(v_1,v_2,\cdots,v_\supsize)=(v_{\col_1,t_1},v_{\col_2,t_2},\cdots,v_{\col_\supsize,t_\supsize})$. Consider the following mechanism $\mainmec$ for value distribution $\dist_{\cols,\ts}$. For each $v_i=v_{\col_i,t_i}$, there is a corresponding menu entry $(x_i,p_i)\in \mainmec$, where $x_i$ allocates to the buyer each set of items $S\in \col_i$ with probability $\frac{1}{m}$; the price of such allocation $x_i$ is $p_i=\frac{t_i}{2}$. 

On one hand, each buyer of type $v_i$ can afford to purchase $(x_i,p_i)$ with utility $\frac{t_i}{2}$. On the other hand, purchasing other lottery $(x_j,p_j)$ will let buyer $v_i$ get an set $S'$ in $\col_i$ with probability at most $\frac{1}{N^{1/4}}$ since $|\col_i\cap\col_j|\leq N^{1/4}$, and in this case his value is $t_i$; otherwise the buyer gets a set $S'$ not in $\col_i$, then since $|S\cap S'|\leq n^{1/4}$, the value he gets is at most $\frac{1}{n^{1/4}}t_i$. Thus the utility of purchasing any other lottery $(x_j,p_j)$ is at most 
\[\frac{1}{N^{1/4}}\cdot t_i+\left(1-\frac{1}{N^{1/4}}\right)\cdot \frac{1}{n^{1/4}}t_i-p_j<0\]
since $p_j=\frac{t_j}{2}\geq\frac{1}{2}$, and $t_i\leq H=\frac{1}{3}n^{1/4}$. Notice that for buyer with type $v_i$, the utility of purchasing any collection of menus is subadditive over the lotteries, since the buyer's valuation function is subadditive. Thus for any adaptive strategy that purchases $\alpha_j$ copies of lottery $(x_j,p_j)$ in expectation, $\forall j\in[\ell]$, the utility of buyer $v_i$ is at most the sum of his utility on purchasing $\alpha_j$ copies of lottery $(x_j,p_j)$ for each $j$. Since his utility for purchasing any other lottery is negative, thus no adaptive strategy purchasing lottery other than $(x_i,p_i)$ can be optimal. Therefore, the optimal strategy for any buyer of type $v_i$ is to purchase one copy of lottery $(x_i,p_i)$. Then $\adaptrev(\dist_{\cols,\ts})=\frac{1}{2\ell}\sum_{i=1}^{\ell}t_i$.

The rest of the proof follows the same flow as the singly-exponential description complexity lower bound proof for a unit-demand buyer. Let $\simplemecs$ be arbitrary class of mechanisms with $|\simplemecs|\leq 2^{\supsize/H^2}$, and set $c=\frac{1}{12}\log H$. It suffices to show that there exists a value vector $\ts$, $\rev_{\dist_{\cols,\ts}}(\mec)<\frac{1}{c}\adaptrev(\dist_{\cols,\ts})$, for every $\mec\in\simplemecs$.

Suppose that we generate a random $\ts=(t_1,t_2,\cdots,t_\supsize)\in[1,H]$, where each $t_i$ is independently drawn from the following truncated geometric distribution: $\Pr[t_i=2^{a}] = \frac{2^{-a}}{1-H^{-1}}$ for $1\leq a\leq \log H$. Fix a mechanism $\mec\in\simplemecs$. We study the probability that $\rev_{\dist_{\cols,\ts}}(\mec)\geq\frac{1}{c}\adaptrev(\dist_{\cols,\ts})$, over the randomness of $\ts$.

For the fixed mechanism $\mec\in\simplemecs$, define $h_i$ to be the payment of a buyer with type $v_{\col_i,t_i}$ in mechanism $\mec$. Thus, to bound the probability that $\rev_{\dist_{\cols,\ts}}(\mec)<\frac{1}{c}\adaptrev(\dist_{\cols,\ts})$, we only need to bound the probability that $\Pr_{\ts}[\frac{1}{\supsize}\sum_i h_i\geq\frac{1}{2c\supsize}\sum_i t_i]=\Pr_{\ts}[\sum_i h_i\geq\frac{1}{2c}\sum_i t_i]$.

For any basic collection $\col_i$, define value distribution $\dist_i$ as follows. A draw $v\sim \dist_i$ can be simulated by first draw $t_i$ according to the truncated geometric distribution: $\Pr[t_i=2^{a}] = \frac{2^{-a}}{1-H^{-1}}$ for $1\leq a\leq \log H$; then return value function $v=v_{\col_i,t_i}$. In other words, for buyer type $v\sim \dist_i$, $v$ is always demanding the same sets, but the value is drawn from some equal-revenue distribution. Notice that selling any lottery to this buyer is equivalent to sell fractional amounts of sets in $\col_i$ to a single-parameter buyer. The optimal revenue from this single-parameter buyer from any mechanism, over the randomness of $t_i$, can be achieved by a pricing the grand bundle of all items at a deterministic price. Such optimum is obtained by selling the grand bundle at price $H$, and the optimal revenue is $\frac{1}{1-H^{-1}}<2$. Thus $h_i$ is a random variable with expectation 
\[\E h_i=\rev_{\dist_i}(\mec)<2.\]
On the other hand, $\E t_i=\frac{\log H}{1-H^{-1}}$. Thus
\begin{eqnarray*}
%\Pr_{\ts}\left[\sum_i h_i \right. & \geq& \left. \frac{1}{2c}\sum_i t_i\right]\\
\Pr_{\ts}\left[\sum_i h_i \geq \frac{1}{2c}\sum_i t_i\right]
&\leq&\Pr\left[\sum_i h_i\geq\frac{1}{4c}\E\left[\sum_i t_i\right]\right]+\Pr\left[\sum_i t_i<\frac{1}{2}\E\left[\sum_i t_i\right]\right]\\
&=&\Pr\left[\sum_i h_i\geq\frac{\supsize}{4c}\cdot\frac{\log H}{1-H^{-1}}\right]+\Pr\left[\sum_i t_i<\frac{\supsize}{2}\cdot\frac{\log H}{1-H^{-1}}\right]\\
&\leq&\exp\left(-\frac{2\left(\frac{\supsize}{4c}\cdot\frac{\log H}{1-H^{-1}}-2\supsize\right)^2}{\supsize H^2}\right)+\exp\left(-\frac{2\left(\frac{\supsize}{2}\cdot\frac{\log H}{1-H^{-1}}\right)^2}{\supsize H^2}\right)\\
&<&2\exp\left(-\frac{2\supsize}{H^2}\right)<2^{-\supsize/H^2}.
\end{eqnarray*}
Here the first line is by union bound. The second line is by $\E t_i=\frac{\log H}{1-H^{-1}}$. The third line is by Hoeffding's inequality and observe that $0\leq h_i,t_i\leq H$. The last line is by $c=\frac{1}{12}\log H$. Then by union bound, since there are only $2^{\supsize/H^2}$ mechanisms in $\simplemecs$, there exists $\ts$ such that no mechanism in $\simplemecs$ can get revenue at least $\frac{1}{c}\adaptrev(\dist_{\cols,\ts})$. This proves Theorem~\ref{thm:lb} since $c=\frac{1}{12}\log H=\Omega(\log n)$, and any mechanism set of size $>2^{\supsize/H^2}$ cannot be described by less than $\frac{\supsize}{H^2}=\frac{9}{\sqrt{n}}2^{2^{n^{1/4}/4}}$ bits.

\end{proof}